\documentclass{article}
\usepackage{ijcai24}

\usepackage{times}
\usepackage{soul}
\usepackage{url}
\usepackage[hidelinks]{hyperref}
\usepackage[utf8]{inputenc}
\usepackage[small]{caption}
\usepackage{graphicx}
\usepackage{amsmath}
\usepackage{amsthm}
\usepackage{booktabs}
\usepackage{algorithm}
\usepackage{algorithmic}
\usepackage[switch]{lineno}
\usepackage{xfrac}

\newtheorem{example}{Example}

\title{Mechanisms that Play a Game,
not Toss a Coin}

\author{
    Toby Walsh
    \affiliations
    AI Institute, UNSW Sydney
    \emails
    tw@cse.unsw.edu.au
}

\begin{document}

\maketitle

\begin{abstract}
Randomized mechanisms can have good normative properties compared
  to their deterministic counterparts. However, randomized mechanisms
  are problematic in several ways such as in their verifiability. We
  propose here to de-randomize such mechanisms by having agents play a
  game instead of tossing a coin. The game is designed so agents
  play randomly, and this play injects ``randomness''  into
  the mechanism. Surprisingly this de-randomization
  retains many of the good normative properties of the original
  randomized mechanism but gives a mechanism that is deterministic and
  easy, for instance, to audit. We consider three general purpose
  methods to de-randomize mechanisms, and apply these to
  six different domains: voting, facility location, task allocation, school choice, peer
  selection,  and resource allocation.
  We propose a number of novel de-randomized mechanisms for these
  six domains with good normative properties (such as equilibria in which
  agents sincerely report preferences
  over the original problem).
  In one domain, we additionally show that a new and
  desirable normative property emerges as a result of de-randomization.
\end{abstract}

\newtheorem{mytheorem}{Theorem}
\newtheorem{myconjecture}{Conjecture}
\newcommand{\myOmit}[1]{}
\newcommand{\myblacksquare}{$\blacksquare$}
\newcommand{\mymin}{\mbox{\rm min}}
\newcommand{\mymod}{\ \mbox{\rm mod} \ }
\newcommand{\myendpoint}{\mbox{\sc EndPoint}\xspace}

\section{Introduction}

In 2012, transplant centres across Germany were placed under criminal
investigation
due to the %systematic
manipulation of donor waiting lists.
Dozens of patients within the Eurotransplant system
had preferentially received organs after doctors falsified the
severity of their illnesses. 
The discovery shattered public trust. % in the system.
In the next year,
donations dropped by around 30\% 
according to the German organ transplant foundation (DSO)
 \cite{2062}.
Unsurprisingly then, when the author was asked to help update
a multi-national mechanism for organ matching, it was made 
clear that the new mechanism should be simple and trustworthy. Indeed, there
was an explicit requirement that it be
deterministic to enable auditing \cite{mswijcai17,mswaies18,mbb}. 

Unfortunately deterministic mechanisms
can struggle to choose between (perhaps essentially equivalent) alternatives
without impacting desirable normative properties like fairness
or strategy proofness (e.g. when breaking ties or resolving 
Condorcet cycles). % or allocating a contested and indivisible item).
Randomization offers an escape since we can  fairly choose
between equivalent alternatives randomly. 
For example, no deterministic mechanism 
for house allocation %with three or more houses
is Pareto efficient, strategy proof and anonymous \cite{rsd}.
%Note that  two of these properties can be achieved
%(e.g. the serial dictatorship mechanism is 
%Pareto efficient and strategy proof but not anonymous).
But randomized mechanisms can have all three properties.
In particular, the random priority mechanism is 
(ex post) Pareto efficient, strategy proof and anonymous.
%As a second example, no deterministic and strategy proof 
%mechanism for locating a facility on the line can do better than 
%2-approximate the maximum cost. However, randomized
%mechanisms exist that can do better, $3/2$-approximating this
%cost \cite{ptacmtec2013}. 

%Even though randomization can lead to good
%normative properties, randomized mechanisms are not 
%universally liked.
There are, however, many
problems with randomization.  First, true rather than pseudo
randomness is difficult to obtain. We typically require access to some
external data source to provide a true source of random bits. 
Second, if we are making what is an one-off decision then
it is inherently difficult to demonstrate that the randomized mechanism 
was fair.
Third, even where decision making is repeated, it can be 
difficult to audit fairness. We may, for instance, require many
executions of the mechanism to have high confidence in its
fairness. Fourth, randomization can introduce undesirable computational
complexity. For instance, randomness can make it computationally 
intractable to compute the probability distribution
over outcomes (e.g. \cite{sswine2013}). 

We propose an alternative %solution
that tackles many of these challenges.
Rather than introduce randomness, 
we stick with deterministic mechanisms but add a 
game where equilibrium behaviour for agents is to 
play randomly. We thereby inject randomness into
a deterministic mechanism through the agents' play.
Surprisingly, we can retain many of the advantages of tossing
a coin while avoiding the disadvantages.
Indeed, de-randomization can even enhance
the normative properties achieved. For instance, we
propose a new peer selection mechanism where
agents always have an incentive to participate,
irrespective of how the other participants vote.

\section{Modular Arithmetic Game}

At the heart of our general purpose method to de-randomize mechanisms
is a simple modular arithmetic game. We first show that equilibrium
behaviour in such games is for two (or more) agents to play randomly. 
The simplest example of such a game is the parity game:
two agents play 0 or 1, the even agent wins if the bits are identical,
while the odd agent wins if they are different. 

\begin{mytheorem}
The parity game has an unique Nash equilibrium in which the even and
odd agent both play uniform random bits. The outcome is even/odd
with equal probability. 
\end{mytheorem}
\begin{proof}
No pure strategy is an equilibrium since one agent must lose and
therefore can profit by deviating and inverting their strategy from odd
to even, or even to odd. 
To show uniqueness of an uniform random strategy,
suppose that the even agent plays odd with probability $p$ and even with
probability $1-p$, and that the
odd agent gets utility $u$ for a win and $v$ for a loss ($u > v$).  To
be an equilibrium,
the odd agent must get the same reward whether they play odd
(expected reward of $(1-p)u + pv$) or even (expected reward of $pu + (1-p)v$)
Therefore $(u-v)(1-2p) = 0$. That is, $p = \sfrac{1}{2}$.
Hence the even agent plays an uniform mixed strategy.
A dual argument holds for the odd agent. 
\end{proof}

In the more general modular arithmetic game, $n$ agents
pick an integer $[0,m)$ where $m \geq n$, and the output is the sum
of these integers mod $m$.
First, we show that if one agent plays an uniform mixed
strategy, picking an integer uniformly at random, then the output of
the game is also an uniform random integer irrespective of the other
agents' strategies. 

\begin{mytheorem}
  If one (or more) agent plays an uniform mixed strategy then the
  outcome of the modular arithmetic game is an uniform random integer 
  irrespective of how the other agents play. 
\end{mytheorem}
\begin{proof}
Let $p(j)$ be the probability that agent 1  picks integer $j$,
and $q(j)$ be the probability that agents 2 to $m$ pick integers that sum
modular $m$ to $j$. Note that $\sum_{i=0}^{m-1} p(i) = \sum_{j=0}^{m-1} q(j) =
1$.
Suppose that $p(j) = \sfrac{1}{m}$. 
The probability that the output of the game is $j$ is
$\sum_{i=0}^{m-1} p(i) q( (j-i) \mymod m)
= \sum_{i=0}^{m-1} \sfrac{1}{m} \  q((j-i) \mymod m)
= \sfrac{1}{m}
\sum_0^{m-1} q((j-i) \mymod m) = \sfrac{1}{m}
\sum_0^{m-1} q(i) = \sfrac{1}{m}$. 
\end{proof}

In modular space, the discrete convolution of any uniform probability distribution
with any other probability distribution is the uniform probability distribution.
Hence, provided one agent plays an uniform random integer, the output
of the modular arithmetic game is itself an uniform random integer.

But are humans any
good at playing randomly? Recent empirical
studies of rock-paper-scissor games 
are hopeful, suggesting that humans can be as random as pseudo-random algorithms
\cite{rsg}:
\begin{quote}
{\em `` \ldots our results demonstrate that human RSG [random sequence generation]
can reach levels statistically indistinguishable from computer
pseudo-random generators in a competitive-game setting \ldots''
}
\end{quote}

We next show that only two
agents need to be playing an uniform random integer for play to be in equilibrium.
We assume that each agent in the modular arithmetic game
has a different most preferred outcome. If each agent
has the same preferred outcome, there are trivial pure Nash equilibria
(e.g. one agent plays the preferred outcome, and all other agents
play 0). On the other hand, if agents have different preferred outcomes, then no pure
strategy is an equilibrium. There are, however, mixed Nash equilibria
with a simple structure. A {\bf quasi-uniform} strategy is a mixed strategy in which
two (or more) agents pick an integer in $[0,m)$ uniformly at random,
and other agents play any mixed or pure strategy. 

\begin{mytheorem}
Any quasi-uniform strategy is a Nash equilibrium which
outputs an uniform random integer in $[0,m)$. 
\end{mytheorem}
\begin{proof}
Suppose agents play a quasi-uniform strategy.  If one agent deviates,
then there is still one other agent playing an uniform mixed
strategy. The outcome therefore remains an uniform random integer, and
the agent's reward is unchanged despite the deviation.
Hence any quasi-uniform strategy is an equilibrium. 
\end{proof}

There are also mixed Nash equilibria that are not
quasi-uniform. Indeed, they are mixed Nash equilibria in which
no agent plays an uniform random integer.

\begin{example}
Consider $n=m=4$. Suppose two agents play 0 or 1 with probability
$\sfrac{1}{2}$,
and the other two agents play 0 or 2 also with probability
$\sfrac{1}{2}$.
Consider one of the agents playing 0 or 1, and another playing 0 or
2. The sum of their plays is an uniform random integer in $[0,3)$.
Hence the sum played by any three of these
agents mod $4$ is an uniform random integer
in $[0,3)$.
Thus any fourth agent has the same reward irrespective of the integer that
they play. This is therefore a mixed Nash equilibrium returning an outcome
that is an uniform random integer in $[0,3)$.
 \end{example}

 Equilibria like those in the last example
 require significant coordination between agents. It is
likely more realistic to suppose that agents will play an uniform or
quasi-uniform strategy as this requires no or limited coordination.

\section{Two Simple Examples}

We illustrate our idea of de-randomizing mechanisms using a modular
arithmetic game with two simple
and classic problems in social choice: voting and facility location.
%We will then consider four more complex domains where
%both the mechanisms and the analysis are
%significantly more complex. These four additional
%domains will be task allocation, school choice, peer selection, and
%resource allocation.

\subsection{Random Dictator}

We first consider one of the most fundamental problems in social choice,
designing mechanisms for voting with good normative problems.
As is well known, we quickly run into a wide range of
impossibility results. For instance,  with three or more candidates,
any voting rule that is surjective and strategy proof must also be
dictatorial \cite{gs1,gs2}. 
One escape from such impossibility results
is randomization. For example, the random dictator mechanism
is surjective and strategy proof but the result is not decided by just
one voter (i.e. it is not dictatorial). Indeed, it is the only
rule that is strategy proof and ex post efficient (i.e. never gives
positive
probability on Pareto-dominated alternatives) \cite{gibbard77}.
However, voters may not be too keen to see
this mechanism being used. The chair might 
tell a voter: ``You preferred candidate lost because
some other voter was randomly chosen as dictator'', and
the voter might have little option but to trust that the chair was
fair. 

We propose instead the following de-randomization of the
random dictator mechanism.
Agents submit an integer in $[0,n)$, along with their preferred
winner. Let $j$ be the sum of these integers mod $n$.
The dictator is chosen to be the preferred winner of the $j+1$th voter. 
This de-randomization of the random dictator
mechanism is not strategy proof. Any voter can ensure
that they are the dictator by choosing a suitable integer. 
However, there is a mixed strategy Nash equilibrium 
in which two (or more) agents choose an integer 
uniformly at random, as well as all agents report sincerely their
most preferred candidate to win. 
A nice feature of this de-randomization (which we will observe in
almost every example explored in this paper) is that an agent's ability to
manipulate the outcome is limited to this modular arithmetic game. It
is in their best interests to declare their preferences over outcomes
(in this case, their preferred winner) 
sincerely.

\subsection{Left-Right-Middle Mechanism}

We turn next to the classic problem of facility location, and show
how we can return better, higher quality solutions using
de-randomization. When locating a single facility on the line, no
deterministic and strategy proof mechanism can do
better than 2-approximate the maximum cost an
agent travels \cite{ptacmtec2013}. However, the
randomized left-right-middle (LRM) mechanism
3/2-approximates the maximum cost, and this
is optimal as no randomized and strategy proof mechanism
can do better. % \cite{ptacmtec2013}.
The LRM mechanism
selects the leftmost agent with probability 1/4,
the midpoint between leftmost and rightmost agents with
probability 1/2, and the rightmost agent again with
probability 1/4.

We propose the following 
de-randomization of the LRM mechanism.
Agents submit an integer between 0 and 3, along with their location. 
If the sum of the integers modulo 4 is 0 then
the facility is located at the leftmost agent.
If the sum of the integers modulo 4 is 1 or 2 then
the facility is located at the midpoint between the leftmost and
rightmost agents. 
Otherwise, the sum of the integers modulo 4 is 3 and 
the facility is located at the rightmost agent. 

This de-randomized facility location
mechanism is not strategy proof. Suppose, for example, that
I am the rightmost agent and I know the other reports. I can ensure the facility is 
located at my location by submitting a suitable integer. However, 
there is a mixed strategy Nash equilibrium 
in which agents choose an integer between 0 and 3 
uniformly at random, as well as reporting their sincere location. The 
expected maximum cost of this mixed
strategy is 3/2 times the optimal maximum cost.
This is better than that obtained by the best deterministic and
strategy proof mechanism. 

These two simple examples have illustrated some of the basic ideas in
de-randomizing randomized mechanisms. We now
apply de-randomization to four other domains where the analysis
is more complex. 
These examples uncover three different but general purpose methods to
de-randomize randomized
mechanisms. In the first (``game-first''), we play a modular arithmetic
game to pick a random ``seed''. This is then applied to the original randomized mechanism.
In the second (``game-last)'', we apply a randomized mechanism to generate a
probabilistic outcome. We then play a modular arithmetic game to convert this into a
discrete ex post outcome. And in the third (``game-interleaved''),
we interleave playing
a modular arithmetic game with applying the mechanism.
%In the first two examples of voting and facility location, 
%we saw a game interleaved method to de-randomize a mechanism.

\section{Task Allocation}

The first more complex domain that we consider is task allocation.
There are $m$ tasks that need to be allocated to 2 agents.
Agent $i$ declares that task $j$ will take time $t^i_j$. The goal
is to allocate tasks to agents to
minimize the makespan (the completion time of the last task
to either agent). To compensate agents for performing a task,
agents receive a payment. The payment is in hours of work, and
the overall utility of an agent is this payment less the actual amount
of work performed. The mechanism design problem here is
to devise a mechanism which approximates well the
optimal makespan and incentivizes agents to report the time it will
take them to execute each task sincerely.

Nisan and Ronen  \shortcite{amd} prove that no deterministic mechanism can
better than 2-approximate the makespan, and that a simple VGC-style
min-work mechanism that allocates each task to the quickest agent,
paying them the time that would be taken by the other agent to perform
the task is strategy proof and
achieves this 2-approximation.
They also prove that randomization can improve this approximation
ratio in expectation. In particular, they show that the biased
min-work mechanism (Algorithm 1) 
provides a 7/4-approximation of the optimal makespan in expectation
when given $m$ random bits drawn uniformly.
This mechanism is strongly truthful (i.e. even if we know the
random bits, it remains strategy proof). We denote these random bits
by $b_j$.

\begin{algorithm}[tb]
  \caption{BiasedMinWork($m,t^i_j,b_j$)}
\begin{algorithmic}[1] %[1] enables line numbers
        \STATE $a_1, a_2 := \{\}$ \mbox{\it \ \ \ \ \ \ \ \ \ \ \ \
          \ \ \ \ \ \ \ \ \ \ \ \ \ \ \ \ \ \ \ \ \ \ \ \ ;
          task allocation}
     \STATE $p_1, p _2 :=  0$ \mbox{\it \ \ \ \ \ \ \ \ \ \ \ \
          \ \ \ \ \ \ \ \ \ \ \ \ \ \ \ \ \ \ \ \ \ \ \ \ \ \  ;
          payment}

            \FOR{$j=1$ to $m$}

    \STATE $i := 1+ b_j, i' := 3-i$   \mbox{\it \ \ \ \ \ \
              \ \ \ \ \ \ \ \  ;
               construct permutation}

            \IF {$t^i_j \leq \frac{4}{3} t^{i'}_j$ }
    \STATE  $a_i := a_i \cup \{j\}, p_i :=  p_i + \frac{4}{3} t^{i'}_j$
    \ELSE
    \STATE $a_{i'} := a_{i'} \cup \{j\}, p_{i'}  :=  p_{i'} + \frac{3}{4} t^{i}_j$
        \ENDIF
\ENDFOR
        \STATE \textbf{return} $a_1, a_2, p_1, p_2$
    \end{algorithmic}
 \end{algorithm}

To de-randomize this biased min-work mechanism, we
have agents set the bits $b_j$ by means of a simple parity game. 
In particular, we suppose the mechanism operates
in $m$ rounds. In the $j$th round, both agents submit their time to
complete task $j$. We additionally suppose they now submit a
single bit, 0 or 1.  We set $b_j$ to be the
xor of these two bits. This example thus illustrates the
``game-interleaved'' method
to de-randomize a mechanism. % where the modular arithmetic
%game is interleaved with the
%original randomized mechanism.

\begin{mytheorem}
  This de-randomized biased min-work mechanism has an unique
  mixed subgame perfect Nash
equilibrium in which agents submit bits uniformly at random, 
and sincerely report their task times. This $\frac{7}{4}$-approximates
the optimal makespan in expectation. 
\end{mytheorem}
\begin{proof}
Consider agent 1 in  round $j$. We say that agent 1 wins
the round iff $b_j = 0$. There are two cases.
In the first, $t^1_j \leq \frac{4}{3} t^{2}_j$.
If agent 1 wins this round (that is, $b_j = 0$), then
agent 1 pays a cost of $t^1_j$ but receives a greater or equal
payment of $\frac{4}{3} t^2_j$.
Agent 2, on the other hand, loses this round,
incurs no cost but receives no
payment. 
Suppose agent 2 instead wins the round (that is, $b_j=1$).
There are two subcases. In the first subcase,
$\frac{3}{4} t^1_j \leq t^2_j \leq \frac{4}{3} t^{1}_j$. Agent 1 is
now not allocated the task and agent 1 therefore has no additional
cost or payment.
However, agent 2 is allocated the task. Agent 2 receives a payment
of $\frac{4}{3} t^1_j$ which is greater than their cost.
Hence, in the first subcase, both agent 1 and 2  want to win the round. 
In the second subcase, 
$t^2_j > \frac{4}{3} t^{1}_j $ and agent 1 is allocated task $j$.
Their payment is $\frac{3}{4} t^2_j$. This is strictly greater than their cost,
$t^i_j$.  However, it is a smaller payment than when agent 1 wins the
round. Hence, it was desirable for both agents to have won this round.
The other case, in which $t^1_j  > \frac{4}{3} t^{2}_j$ is dual.
The mechanism treats agents identically so agent 2 
also wants to win each round. 
It follows that the unique mixed subgame perfect Nash equilibrium has
both agents submitting bits uniformly at random. The biased min-work
mechanism is strongly truthful as, even if the result of the parity game
is known, agents  have no incentive to misreport their task time.
As the subgame perfect Nash equilibrium has agents winning
each round with equal probability, the mechanism returns the same
distribution of task allocations
as the randomized mechanism. Hence, it 
$\frac{7}{4}$-approximates the optimal makespan in expectation. 
\end{proof}

\section{Peer Selection}

We consider next the peer selection problem in which agents
have to choose one amongst themselves to receive a prize (e.g.
\cite{selectingselectors,nsf1,spsmae,spsrpa,wecai2014}).
For example, a committee might want to choose one person
within the committee to be chair.
As a second example, a school class might want to choose a class
representative.
As a third example, the members of an academic society might want
to choose a member to receive a medal.
We propose here a new peer selection mechanism with a novel
normative property encouraging participation: each one of the agents is guaranteed
to have a say in the winner. That is, irrespective of how
the other agents vote, every agent can change the winner of the prize.
This peer selection
mechanism is based on the idea of sequentially eliminating
agents until just one is left who wins the prize.
Our results
\myOmit{ about
this sequential elimination mechanism}
easily generalize
\myOmit{from the peer selection setting where
  $n$ agents vote over (the same) $n$ candidates}
to the more general voting
setting where $n$ agents vote over $m$ candidates. % ($m \neq n$).
%and $m$ may be different to $n$. 

We first consider a deterministic peer selection mechanism in
which agents successively eliminate candidates. This mechanism
lacks the desirable normative property of anonymity (i.e. permuting
the names of the agents changes the outcome). 
We therefore randomize it to give anonymity. We will then de-randomized
this mechanism to give a new deterministic and anonymous
mechanism. This de-randomized mechanism has
the property that, even if some agent has a large majority support, a
dissenting agent still has an incentive to participate and change the
outcome. This example
is an instance of the ``game-first'' method for de-randomizing a randomized
mechanism in which we play a modular arithmetic game first to select a random
``seed''. 

\subsection{Sequential Elimination}

We start by considering the strategy behaviour of
agents with the simple deterministic sequential elimination (SE)
mechanism
studied by Moulin \shortcite{moulinbook} and others.
We will need this equilibrium result to discuss the equilibria of the de-randomized
mechanism that we will shortly introduce. The SE
mechanism starts with all agents as possible winners,
and then has the agents in a given order
eliminate one of the remaining agents from the set of possible
winners until only a single agent remains. 
We suppose the $n$ agents have strict preferences
so there is, for example, always an unique least preferred agent to
eliminate.
This SE mechanism is in some sense a dual of the
dictatorship mechanism.
%In the dictatorship mechanism,
%a particular agents picks an agent as winner, whilst in the SE
%mechanism, particular agents eliminate agents as winner. 
Whilst the dictatorship mechanism is strategy proof, sequential elimination is not.
A agent may not eliminate their worst ranked agent
if a later agent will. Strategic play is, however, easily computed.

Given a fixed elimination ordering,
we can view the SE mechanism as a
repeated game. Supposing agents have complete information about
the elimination ordering and other preferences,
the subgame perfect Nash equilibrium of
this game can be computed using backward induction. 
However, such a computation is exponential in $m$. 
There is, however, a simple linear time method to compute
the subgame perfect Nash equilibrium. We simply reverse the
elimination ordering and play the game backwards
\cite{moulinbook}. This is 
similar to equilibria of the sequential
allocation mechanism itself
\cite{knwxaaai13,waaai16,agwadt2017}.

It is easy to see informally why this is the case. 
Agents can play strategically to ensure that the last agent
to be eliminated is the least preferred agent of the
last agent in the elimination ordering.
The last agent will surely want this agent to be eliminated.
Therefore earlier agents in the elimination order
might strategically not eliminate this agent, even if it
is also their least preferred agent. An early agent
will profit by eliminating some more preferred agent, safe
in the knowledge that their least preferred agent will be
eliminated at the end of the game. 
Similarly, if we
discount this agent, the agents can play strategically to
ensure that the penultimate agent
to be eliminated is the least preferred agent of the
penultimate agent in the elimination ordering, and so on. 
Indeed, all subgame perfect Nash equilibria return the same
winner, the one computed by this reversed computation.

\myOmit{
\begin{mytheorem}
  All subgame perfect Nash equilibria of the SE mechanism
game given the picking ordering $\pi$ of 
agents have the same winner, and this is
found in $O(m)$ time by computing $elim(rev(\pi))$ where $rev$
reverses the elimination ordering, and 
$elim$ has the agents in turn eliminating their least preferred
remaining agent.
\end{mytheorem}
\begin{proof}
We prove a strictly more general result considering the game where we eliminate $k$
agents to leave a set of $m-k$ winners. To compare outcomes with multiple
winners (i.e. $m-k>1$),
we suppose agents have additive utilities over sets of agents. 
To show that the winners of the subgame perfect Nash equilibrium
of this SE game are computed by
$elim(rev(\pi))$, we use induction on the number of  elimination steps
$k$. In the base case, $k=1$ and the 
optimal play for the single agent
that eliminates a single agent is computed by $elim(rev(\pi))$. 
In the step case, we suppose the induction hypothesis holds
for $k$ ($k \geq 1$) elimination steps and demonstrate that it holds
for $k+1$ such steps.
Consider the last agent eliminated in the final $k+1$th step.
Suppose the last agent eliminated
is not the least preferred agent of the agent
eliminating agents at this step. There are two cases.
In the first case, the least preferred agent survives to be in the
winning set. The final step
would be dominated by eliminating this least 
preferred agent contradicting this being subgame perfect.
In the second case, the least preferred agent
has already been eliminated before the final step.
We can then swap the final elimination with the elimination
of this least preferred agent and construct an equivalent
game in which the final elimination is the least preferred agent
of this agent. We then appeal to the induction hypothesis. 

To demonstrate that the winner is unique, suppose there 
are two subgame perfect Nash equilibria with
different winners. Consider the agents in common
eliminated in both games. Without loss of generality,
we permute the eliminations in one game so that these are
in the same order as in the second game. Consider next the first agent in the
first game that eliminates the winner of the second 
game. Recall that we assume agents have
strict preferences so this agent will strictly prefer
one of these two different outcomes. But this means
one of these eliminations is dominated which contradicts that
this is a subgame perfect Nash equilibrium. 
Note that there can
be multiple subgame perfect Nash equilibria
that give the same winner.
For example, if in one equilibrium,
agent 1 eliminates 
agent $a$, and agent 2 eliminates agent
$b$, and both agents 1 and 2 have agents $b$
and then $a$ at the end of their preference ordering,
then there is another equilibrium with the
same winner in
which agent 1 eliminates agent $b$ and
agent 2 eliminates agent $a$. 
\end{proof}

Note that this result about computing the subgame perfect
Nash equilibrium of the SE mechanism also
holds for the more general voting game in which the number of
candidates
may be different to the number of voters, and not just to the peer
selection game in which it is equal. 
We also contrast this result
with the unique subgame perfect Nash
equilibrium of the sequential allocation mechanism in the
fair division game with indivisible items \cite{knwxaaai13,kcor71}.
With two agents, this can also be computed by reversing the picking sequence.
However, the computation of the subgame perfect Nash
equilibrium in the fair division game also reverses the preferences so that
agents are allocated their least preferred unallocated items
remaining rather than their most
preferred. In this fair division game,
agents strategically pick to ensure that
the last remaining item is the last agent's least preferred item. 
And if we discount this item, the last but
one remaining item is that agent's least preferred item, and so on.
}

\subsection{Random Sequential Elimination}

The SE mechanism that we just considered
is not anonymous. It treats
agents differently according to where they appear in the elimination
ordering. We can make it anonymous by randomizing over the agents.
In particular, we propose the random sequential elimination (RSE) mechanism
which picks a random order of agents and then has agents eliminate
agents one by one using this 
ordering until just one
agent remains who is declared the winner. 
%This RSE
%mechanism is a dual of the random dictator mechanism which has
%a random agent select (rather than eliminate) the winner.

We can de-randomize RSE using a ``game-first'' method. 
In the first stage, the $n$ agents play a modular arithmetic game in which 
each submits an integer in $[0,n!)$. We use the result of this game to pick 
one of the $n!$ priority orderings of agents
via an inverse Lehmer code. In the second stage, 
we eliminate agents using this ordering. 
With this
de-randomized RSE mechanism, agents have only limited options for strategic
play. In particular, 
the voting game for the de-randomized mechanism has a mixed subgame
perfect Nash equilibrium in which agents pick random integers
uniformly,  and then eliminate their least preferred remaining
agent in reverse order.

\begin{mytheorem}
  The voting game of the de-randomized RSE mechanism has a mixed subgame perfect Nash
  equilibrium in which two or more agents pick numbers uniformly at random in the
  first stage, and then agents sincerely eliminate their least preferred agent remaining
  in the reverse elimination order in the second stage. 
\end{mytheorem}
\myOmit{
  \begin{proof}
Suppose all but one agent plays this uniform mixed strategy. 
Then, in the first stage, a agent gets the same expected return irrespective of  
the integer that they play. Note that, once the first stage has
been played, the elimination order for agents is fixed
and the subgame perfect Nash equilibrium for the second stage
is to pick the least preferred remaining agent with the reverse
elimination order. 
\end{proof}
}

A novel property of this mechanism is that agents always have an incentive
to participate irrespective of the votes of the other agents.
We say that a peer selection mechanism over two or more
agents is {\bf responsive} iff, regardless of the votes of the other agents, 
there exist two different ways that any agent can vote with
different outcomes. Note that the SE mechanism is not responsive
as the last agent in the elimination order cannot change the outcome.
The de-randomized RSE mechanism is, on the other hand, responsive. 

\begin{mytheorem}
  The de-randomized RSE mechanism is responsive. 
\end{mytheorem}
\begin{proof}
Consider the first elimination round and any agent. There are two cases.
In the first case, the agent is chosen to perform the first elimination. By
changing the agent that is eliminated to be
the current winner, the overall winner must change.
In the second case, the agent is not chosen to perform the first elimination. Suppose
this agent submits a different integer to ensure that they perform the
first elimination. If the agent now
eliminates the current winner, then the overall winner must again change. 
\end{proof}

Another desirable normative property in peer selection is
impartiality \cite{hmeconometrica2013}.
A peer selection mechanism is {\bf impartial} iff
an agent cannot influence whether they win or not. 
For example, we can design an impartial mechanism by
partitioning agents into two sets, having each set vote on the other, and then
randomly choosing between the two possible winners.
We can de-randomize this mechanism without violating impartiality
by having the agents who are not one of the two possible winners
playing a parity game to pick between the two possible winners.
The resulting ``game-last'' mechanism
is impartial but is not responsive. On the other hand, the de-randomized RSE
mechanism is responsive but not impartial. This is to be expected as 
impartiality and responsiveness are impossible to
achieve simultaneously. 

\begin{mytheorem}
No deterministic peer selection mechanism is both responsive and impartial. 
\end{mytheorem}
\begin{proof}
With two agents, the only impartial mechanism selects
a fixed winner. This is not responsive. Consider then three or more
agents, and any reports for these agents. Pick any agent. Suppose this
agent wins. If the mechanism
is responsive, there must be some other report for this agent that changes the winner.
Consider the agent changing from this new report to its original
report. This change violates impartiality. This is a contradiction.
Hence, the assumption that the agent  wins is false. But this is
true for every possible agent. Thus, no agent can win. 
\end{proof}

We end this section with some important related work.
Bouveret {\it et al.}
\shortcite{seqelim} previously studied the SE mechanism for the more general
voting setting. 
They argue that the SE mechanism has low-communication complexity
and, with a good elimination ordering, can have good
properties such as returning the Borda winner.
Here we considered such a mechanism where
the elimination ordering is chosen randomly. The resulting
randomized mechanism is now anonymous (as is the de-randomized version).

\section{School Choice}

We turn next to a very practical social choice problem that impacts
citizens of many different countries.  In school choice, we consider
the two sided matching problem in which schools and students choose
each other (e.g. \cite{Abdulkadiroglu03:School,nashbm,boston}). 
One of the most popular mechanisms for school choice in every day use
around the world
is the student-proposing deferred acceptance (DA) mechanism of %Gale and Shapley
\cite{galeshapley}. This % mechanism
has many desirable normative properties
such as strategy proofness for the students (but not schools) and
stability of the solution.

One issue with the DA mechanism is that it
supposes schools have a complete ordering over students. In practice,
schools often only have broad preferences over students.
For example, those with siblings at the school might be strictly preferred to
those merely within the school district, and these two groups might be
strictly preferred to those outside the school district. However,
within each group, schools might be indifferent between students. It is therefore
common to order students within each group using a random lottery. 

To de-randomize the DA mechanism which uses such a random lottery, 
we could have the $n$ students in a particular group
instead play a modular arithmetic game by submitting an integer in
$[0,n!)$. We then use the result of this game to pick 
one of the $n!$ priority
orderings of students in this group via an inverse Lehmer code.
This school choice example thus illustrates the game-first method
to de-randomize a randomized mechanism: we first play a modular arithmetic game
to construct a random ``seed'' (ordering) which is then used in the
second step by the original mechanism. 

While this game is polynomial, requiring students
to submit just $O( n \log n)$ bits, it may nevertheless be
prohibitively expensive. For instance, with 1000 students to
order in a group, we would require a student to submit
an integer with several thousand decimal digits (as $1000! \approx 4.0$ x $10^{2568}$).
We propose a more efficient mechanism where, instead of each student
submitting $O(n \log n)$ bits, each of the $n$ students submits only $O(\log n)$
bits. These are then combined to form a (random) $O(n \log n)$ priority
order. The mechanism differs on whether $n$ is odd or even. We suppose
students are numbered from 0 to $n-1$. 

If $n=2k$, student 0 submits an integer $b_0 \in [0,n)$,
each student $i \in [1,n-1)$ submits two integers, $a_i \in [0,i]$ and
$b_i \in [0,n-i]$, and the final student $n-1$ submits an integer $a_{n-1} \in
[0,n)$. We construct a permutation ordering as follows.
The first student in the ordering is $(a_{n-1} + b_0) \mymod n$.
The second student is then the $(a_{n-2} + b_1) \mymod (n-1)$ largest
remaining student (counting from zero). 
The third student in the ordering
is then the $(a_{n-3} + b_2) \mymod (n-2)$ largest
remaining student, and so on.

If $n=2k+1$, 
student 0 submits two integers $a_0 \in [0,k]$ and $b_0 \in [0,n)$,
each student $i \in [1,n-1)$ submits two integers, $a_i \in [0,i]$ and
$b_i \in [0,n-i]$, and the final student $n-1$ submits two integers $a_{n-1} \in
[0,n)$ and $b_{n-1} \in [0,k]$. We construct a permutation from these as follows:
the first student in the permutation is $(a_{n-1} + b_0) \mymod n$.
The second student is then the $(a_{n-2} + b_1) \mymod (n-1)$ largest
remaining student (counting from zero), and so on.
There is, however, one exception in the exact middle of the
permutation as the two integers being added together would otherwise
be submitted by the
same agent. More precisely,
the $k+1$th student in the permutation is computed not as the
$(a_k + b_k) \mymod (k+1)$ largest remaining student
but instead as the $(a_k + b_k + a_0 + b_{n-1}) \mymod (k+1)$
largest. 

In the second stage of our de-randomized mechanism, we
run the regular DA mechanism using this
priority order to break ties that a school has within a group.
This de-randomized mechanism inherits the stability
of the underlying DA mechanism. In addition, whilst students can
act strategically, their strategic behaviours are limited to how they
play the modular arithmetic game. 

\begin{mytheorem}
  The de-randomized DA mechanism has
  a mixed Nash equilibrium in which two or more students 
  in the first stage select integers with uniform probability,
  and then all students select schools in the second stage sincerely.
  This equilibrium corresponds to a probability distribution over ex post stable
  matchings. 
\end{mytheorem}
\myOmit{\begin{proof}
Suppose all but one student plays this uniform mixed strategy  
and picks sincerely in the second stage.  
Then, this one student gets the same expected return irrespective of  
the integers that they play. Note that, once the priority ordering has been  
selected, students have, as in the regular DA mechanism,
no incentive to misreport.  
Hence, this uniform mixed strategy is a Nash equilibrium and is a
distribution over ex post stable matchings. 
\end{proof}}

\section{Resource Allocation}

Our final application domain is resource allocation where we consider
(randomized) mechanisms for the allocation of indivisible items \cite{wijcai2020}.
We have $n$ agents who have preferences over $m$ indivisible items.
Our goal is to allocate items whole to agents according to these
preferences. 
In this setting, randomization allows us to deal with contested items
fairly. For example, rather than unilaterally give a contested item to one agent, we could toss a
coin to decide which agent gets it. 
Two of the most prominent randomized mechanisms with good normative
properties for this domain %for the indivisible resource allocation problem 
are the probabilistic serial (PS) and random priority (RP) mechanisms.

\subsection{Probabilistic Serial}

In the probabilistic serial (PS) mechanism agents simultaneously
``eat'' at a constant speed their most preferred remaining item
\cite{psm}.
This gives a randomized or probabilistic 
assignment which can easily be realized as a 
probability distribution over discrete
allocations. Unfortunately, the PS mechanism
is not strategy proof (see \cite{agmmnwaamas2015,agmmnwijcai2015}). However, it has good 
welfare and efficiency properties. It is, for instance,
SD-efficient and SD-envyfree\footnote{These are efficiency
  and fairness notions that are defined for ordinal preferences.
  In particular, the {\em SD} (stochastic dominance)
ordering prefers an allocation $p$ to an agent over
$q$ if the probability for the agent to get the top $i$ items
in $p$ is at least as large
as in $q$ for all $i \in [1,m]$. 
If an allocation to an agent is {\em SD-preferred} over
another then it has greater or equal
expected utility for all utilities consistent
with the agent's ordinal preferences. 
Notions like 
{\em SD-efficiency}, {\em SD-envy freeness} and {\em SD-strategy
proofness} can be defined in the natural way from the SD-preference ordering.}.

To de-randomize the PS mechanism, we first identify how
much precision is needed to represent the probabilities in the
random
assignment it generates. This precision will dictate the size of the
modular aritmetic game that agents will play. 
We prove here that the PS mechanism only needs
a polynomial number of bits with which to represent probabilities.
\begin{mytheorem}
  For every one of the $n$ agents and of the $m$ items,
  there exists an integer $k$ such
  that the PS
  mechanism allocates the item to the agent with probability
  $\frac{k}{(n!)^m}$.
\end{mytheorem}
\begin{proof}
We suppose without loss of generality that items are
fully ``consumed'' by the PS mechanism
in numerical order. We can consider then $m$ steps of the
PS mechanism, each ending with a new item
being fully consumed. We assume that only one item is ever
fully consumed at a time. We will discuss relaxing this assumption
shortly. 

Let $k_{i,j}$ be the number of agents eating the $i$th item at the
$j$th step of the probabilistic serial mechanism. For notational
simplicity, we write $k_i$ for $k_{i,i}$.
Hence, $k_{i}$ is the number of agents eating the $i$th item when it is fully
consumed. Note that $k_i \in [1,n]$ and $k_{i,j} \in [0,n)$ for $i \neq j$.

The first step takes $\frac{1}{k_1}$ time for the first item to
be fully consumed. The $k_1$ agents each get a share of
$\frac{1}{k_1}$ of probability of this first fully consumed item. Note that this is an
integer multiple of $\frac{1}{n!}$, and thus also of
$\frac{1}{(n!)^m}$. Consider any item $j > 1$.
The $k_{j,1}$ agents eating this item each get a share of
$\frac{1}{k_1}$ of probability of item $j$.  Note that this is again an
integer multiple $\frac{1}{(n!)^m}$. There is now
$1-\frac{k_{j,1}}{k_1}$
of item $j$ left. That is, $\frac{k_1-k_{j,1}}{k_1}$ of item $j$
left. Note this is an integer multiple of $\frac{1}{n!}$

Supposing $m\geq 2$, during the next step, $k_2$ agents
eat whatever remains of the second item. This takes
$(\frac{k_1- k_{2,1}}{k_1})/k_2$ time. Each of the $k_2$
agents eating this item thereby receives a share of
$\frac{k_1-k_{2,1}}{k_1 k_2}$ of probability of item 2.
Note this is an integer multiple of $\frac{1}{(n!)^2}$, and thus
of $\frac{1}{(n!)^m}$.
Consider any other item $j > 2$.
The $k_{j,1}$ agents eating this item each get a share of
$\frac{k_1-k_{2,1}}{k_1 k_2}$ of probability of item $j$.
Note that this is again an
integer multiple $\frac{1}{(n!)^m}$. There is now
$\frac{k_1-k_{j,1}}{k_1} - \frac{k_{j,1}(k_1-k_{2,1})}{k_1 k_2}$ 
of item $j$ left.  That is, $\frac{ k_1 k_2 - k_{j,1} (k_1 + k_2) +
  k_{j,1} k_{2,1}}{k_1 k_2}$. Note that this is an integer multiple of
$\frac{1}{(n!)^2}$.

The argument repeats in subsequent steps. In the $j$th step, agents add
an additional probability for an item which is an integer multiple
of $\frac{1}{(n!)^j}$. And the amount left of any item not fully
consumed is also an integer multiple of $\frac{1}{(n!)^j}$.
Adding together all the probabilities, we conclude at the final $m$th step
that for each agent and item, there exists an integer $k$ such
that the PS mechanism allocates the item to the agent with probability
$\frac{k}{(n!)^m}$.
Note that if two or more items are fully consumed at exactly the same time,
the argument is similar. However, we now have strictly less than $m$
steps but the conclusion still holds. 
Note that $(n!)^m$ can be represented in $O(mn \log(n))$ bits which
is polynomial in $n$ and $m$.
\end{proof}

We next define a two stage mechanism that de-randomizes the
PS mechanism. In the first stage,
the mechanism constructs the probabilistic allocation of the usual
PS mechanism. In the second stage, agents play a modular arithmetic
game, and the mechanism uses the outcome of this game
to construct a discrete ex post outcome. This
illustrates the ``game-last'' method to de-randomize a randomized mechanism.
%We apply a randomized mechanism to construct a probabilistic
%allocation, and then play a modular arithmetic game to convert this into a discrete ex
%post outcome. 

In more detail, in the second stage,
each agent submits an integer in $[0,(n!)^m)$. Let $\sigma$
be the sum of these integers mod $(n!)^m$. The mechanism then allocates
item $j$ to agent $\mymin \{  k \ | \ \sum_{i=1}^k p_{i,j} \geq
\frac{\sigma}{(n!)^m}\}$ where $p_{i,j}$ is the probability
that agent $i$ is allocated item $j$ in the probabilistic allocation. 
In other words, the mechanism treats $\sigma / (n!)^m$ as a
random draw. Note that this is a worst case. As this is a ``game-last'' de-randomization,
we can analyse the
probabilistic allocation returned in the first stage and perhaps use
a smaller range of integers than $[0,(n!)^m)$. 
%In the second stage of the de-randomized PS mechanism,
%irrespective of the play in the first stage,
%
%Finally we prove that
The second stage has a mixed Nash equilibrium
in which agents choose an integer uniformly at random. This gives
each agent an expected welfare equal to that of the randomized
allocation computed in the first stage. 

\begin{mytheorem}
  A mixed Nash equilibrium of the second stage has two or more agents
  select integers in $[0,(n!)^m)$ with uniform probability. 
\end{mytheorem}
\myOmit{
\begin{proof}
Suppose all but the first agent play this uniform mixed strategy.
% (i.e. selecting each $i \in [0,(n!)^m)$ with equal probability).
Then,
the first agent gets the same expected return irrespective of
the integer %$j \in [0,(n!)^m)$ that
they play. Hence, this %the uniform mixed strategy
is a Nash equilibrium. 
\end{proof}
}

As with the PS mechanism itself, agents may strategically misreport
their ordinal preferences over items in the first stage of this
de-randomized PS mechanism. For example if one agent prefers
item 1 to 2 and 2 to 3, and a second agent prefers 2 to 3 and 3 to 1,
then the first agent can profitably
misreport that they prefer 2 to 1 and 1 to 3.
A pure Nash equilibrium of the PS mechanism
is guaranteed to exist, but is 
NP-hard to compute in general \cite{agmmnwijcai2015}. 
Indeed, even computing a best response is NP-hard
\cite{agmmnwaamas2015}.

One tractable special case is two agents since
there exists a linear time
algorithm to compute a pure Nash equilibrium with two agents which
yields the same probabilistic allocation as the truthful
profile \cite{agmmnwijcai2015}.
As the PS mechanism is envy-free ex ante, it follows
that a mixed Nash equilibrium with two
agents consisting of this pure strategy for the ordinal preferences and
an uniform mixed strategy for the modular arithmetic game is envy-free ex ante. 
Another tractable special case is 
identical ordinal preferences. The PS mechanism is strategy proof in
this case.  Hence, 
a combined mixed Nash equilibrium with identical ordinal preferences
(consisting of a sincere strategy for the ordinal preferences and
an uniform mixed strategy for the modular arithmetic game) is again
envy-free ex ante.

\subsection{Random Priority}

For simplicity, we consider house allocation (i.e. one item
to allocate per agent). However, our results easily generalize
to a setting with more (or fewer) items than agents.
The random priority (RP) mechanism
picks a random order of agents, and then agents
take turns according to this order to pick their most
preferred remaining item (house)
(e.g. \cite{Abdulkadiroglu98:Random,rsd}).

Random priority is used in many real world settings.
For example, the author was allocated undergraduate housing
using the RP mechanism.
RP is one of the few strategy proof mechanisms available
for house allocation. Indeed, any strategy proof, nonbossy\footnote{A mechanism is {\em nonbossy} if an agent
cannot change the allocation without changing their
own allocation.}
and
neutral\footnote{A mechanism is {\em neutral} if 
permuting the names of the items merely permutes the outcome}
mechanism is necessarily some form of serial
dictatorship like the RP mechanism \cite{Svensson99:Strategy}.

We define a two stage ``game-first'' mechanism that de-randomizes the
RP mechanism. In the first stage,
agents play a modular arithmetic
game, and the mechanism uses the outcome of this game
to construct a picking ordering for the second stage. 
Each agent submits an integer in $[0,n!)$. We 
sum these integers mod $n!$ and then convert this
into a permutation ordering of agents via an inverse Lehmer code.
In the second stage, the mechanism uses this ordering
with a serial dictatorship mechanism to allocate items to agents.
This is then another example of a ``game-first'' method to de-randomize
a randomized mechanism.

\begin{mytheorem}
  The de-randomized RP mechanism has
  a mixed Nash equilibrium in which two or more agents 
  in the first stage select integers in $[0,n!)$ with uniform probability,
  and all agents then select items in the second stage sincerely. 
\end{mytheorem}
\myOmit{
\begin{proof}
Suppose all but the first agent play this uniform mixed strategy.
%(i.e. selecting each
%$i \in [0,n!)$ with equal probability and then picking sincerely).
Then, the first agent gets the same expected return irrespective of
whatever integer % $j \in [0, n!)$ that
they play. Once the ordering has been
selected, agents have no incentive to misreport.
Hence, %the uniform mixed strategy
this is a Nash equilibrium.
\end{proof}
}

Like RP which returns a probability distribution over
ex post outcomes that are Pareto efficient, the de-randomized
RP mechanism is Pareto efficient. 
Note that we could play the modular arithmetic game second,
after we first compute the randomized allocation returned by the
RP mechanism, and then use the outcome of the modular arithmetic game to perform
a ``random draw'' like we did when de-randomizing the PS mechanism.
However, it is NP-hard to compute the
probabilistic allocation returned by RP \cite{saban}. Therefore
this is not a tractable de-randomization.

\section{Conclusions}

We have proposed three related methods of de-randomizing
mechanisms: ``game-first'', ``game-last'' and
``game-interleaved''. Each introduces a modular arithmetic game
which can inject randomness into the mechanism via the randomness in
agents' play. Surprisingly, these de-randomized mechanisms retain many of the good
normative properties of the underlying randomized mechanism. For
instance, all but one of the de-randomized mechanisms
is ``quasi-strategy proof'' as they have mixed Nash
equilibria in which agents play the game randomly but report other
preferences sincerely. We demonstrate how these de-randomization
methods work
in six different domains: 
voting, facility location, task allocation, school choice, peer
selection,  and resource allocation. In one case, de-randomization
additionally introduced a new and desirable normative property (namely
that the
de-randomized peer selection mechanism was 
responsive always to the preferences of any agent).

\section*{Acknowledgements}

The author is supported by the Australian Research Council
through an ARC Laureate Fellowship FL200100204. 

\bibliographystyle{named}
%\bibliography{/Users/z3193295/Documents/biblio/a-z2,/Users/z3193295/Documents/biblio/pub2}

%% The file named.bst is a bibliography style file for BibTeX 0.99c

\end{document}